\documentclass[11pt]{article}
\usepackage[margin=1in]{geometry}
\usepackage{graphicx}
\usepackage{natbib}
\usepackage{footnote,enumerate,amsmath,amssymb,amsfonts,amsthm,subcaption,hyperref}

\usepackage[linesnumbered,ruled]{algorithm2e}
\newtheorem{theorem}{Theorem}
\newtheorem{proposition}[theorem]{Proposition}

\begin{document}
\title{Finding Maximum Cliques on the D-Wave Quantum Annealer}
\author{Guillaume Chapuis\\Los Alamos National Laboratory
\and Hristo Djidjev (PI)\\Los Alamos National Laboratory
\and Georg Hahn\\Imperial College, London, UK
\and Guillaume Rizk\\INRIA/Irisa, Rennes Cedex, France}
\date{}
\maketitle

\begin{abstract}
This paper assesses the performance of the D-Wave 2X (DW) quantum annealer for finding a maximum clique in a graph, one of the most fundamental and important NP-hard problems. Because the size of the largest graphs DW can directly solve is quite small (usually around 45 vertices), we also consider decomposition algorithms intended for larger graphs and analyze their performance. For smaller graphs that fit DW, we provide formulations of the maximum clique problem as a quadratic unconstrained binary optimization (QUBO) problem, which is one of the two input types (together with the Ising model) acceptable by the machine, and compare several quantum implementations to current classical algorithms such as simulated annealing, Gurobi, and third-party clique finding heuristics. We further estimate the contributions of the  quantum phase of the quantum annealer and the classical post-processing phase typically used to enhance each solution returned by DW. We demonstrate that on random graphs that fit DW, no quantum speedup can be observed compared with the classical algorithms. On the other hand, for instances specifically designed to fit well the DW qubit interconnection network, we observe substantial speed-ups in computing time over classical approaches.
\end{abstract}

\section{Introduction}
\label{section_introduction}
The emergence of the first commercially available quantum computers by D-Wave Systems, Inc.~\citep{dwave2016} has provided researchers with a new tool to tackle NP-hard problems for which presently, no classical polynomial-time algorithms are known to exist and which can hence only be solved approximately (with the exception of very small instances which can be solved exactly).

One such computer is D-Wave 2X, which we denote here as DW. It has roughly $1000$ units storing quantum information, called \textit{qubits}, which are implemented via a series of superconducting loops on the DW chip. Each loop encodes both a 0 and 1 (or, alternatively, -1 and +1) value at the same time through two superimposed currents in both clockwise and counter-clockwise directions until the annealing process has been completed and the system turns classical \citep{Johnson2011,Bunyk2014}.

The device is designed to minimize an unconstrained objective function consisting of a sum of linear and quadratic binary contributions, weighted by given constants. Specifically, it aims at minimizing the \textit{Hamiltonian}
\begin{equation}
  H=H(x_1,\dots,x_N)=\sum_{i \in V} a_i x_i + \sum_{(i,j) \in E} a_{ij} x_i x_j \label{eq:hamilt}
\end{equation}
with variables $x_i\in\{0,1\}$ and  coefficients $a_i$, $a_{ij} \in \mathbb{R}$, where $V = \{1,\ldots,N\}$ and $E=V \times V$ \citep{King2015}. This type of problem is known as a \textit{quadratic unconstrained binary optimization (QUBO)} problem. When the coefficients $a_i$ and $a_{ij}$ are encoded as capacities of the \textit{couplers} (links) connecting the qubits, $H$ describes the quantum energy of the system: During annealing, the quantum system consisting of the qubits and couplers tries to settle in its stable state, which is one of a minimum energy, i.e., of a minimum value of $H$. In order to solve a given optimization problem, one has to encode it as a minimization problem of a Hamiltonian of type \eqref{eq:hamilt}. 

Similarly to the random moves considered in a simulated annealing classical algorithm, a quantum annealer uses quantum tunneling to escape local minima and to find a low-energy configuration of a physical system (e.g., constructed from an optimization problem). Its use of quantum superposition of $0$ and $1$ qubit values enables a quantum computer to consider and manipulate all combinations of variable values simultaneously,  while its use of quantum tunneling allows it to avoid hill climbing, thus giving it a potential advantage over a classical computer. However, it is unclear if this potential is realized by the current quantum computing technology, and by the DW computer in particular, and whether DW provides any quantum advantage over the best available classical algorithms \citep{Ronnow2014,Denchev2016}.

This article tries to answer these questions for the problem of finding a maximum clique (MC) in a graph, an important NP-hard problem with multiple applications including network analysis, bioinformatics, and computational chemistry. Given an undirected graph  $G = (V, E)$, a \textit{clique} is a subset $S$ of the vertices forming a complete subgraph, meaning that any two vertices of $S$ are connected by an edge in $G$. The clique size is the number of vertices in $S$, and the maximum clique problem is to find a clique with a maximum number of vertices in $G$ \citep{Balas1986}. 

We will consider formulations of MC as a QUBO problem and study its implementations on DW using different tools and strategies. We will compare these implementations to several classical algorithms on different graphs and try to determine whether DW offers any quantum advantage, observed as a speedup over classical approaches.

The article is organized as follows. Section~\ref{section_relatedwork} starts with an overview of related work that aims to solve graph and combinatorial problems with DW, in particular previous work on the maximum clique problem. Section~\ref{section_methods} proceeds by introducing the qubit architecture on the DW chip as well as available software tools. We also describe a QUBO formulation of MC together with its implementations on DW and present methods for dealing with graphs of sizes too large to fit onto the DW chip. Section~\ref{sect:experiments} presents an experimental analysis of the quantum software tools and a comparison with several classical algorithms, both for graphs small enough to fit DW directly as well as for larger graphs for which decomposition approaches are needed. We conclude with a discussion of our results in Section~\ref{section_conclusion}.

In the rest of the paper, we denote a graph as $G=(V,E)$, where $V = \{ 1,\ldots,N \}$ is a set of $N \in \mathbb{N}$ vertices and $E$ is a set of undirected edges.

\section{Related work}
\label{section_relatedwork}
Several publications available in the literature aim at searching for a quantum advantage within a variety of problem classes. Existing publications often target a particular (NP-complete) problem and compare the performance of a quantum annealer (by D-Wave Systems, Inc.~\citep{dwave2016}) to state-of-the-art classical or heuristic solvers. Early examples include \textit{multiple query optimization} in databases, analyzed by \cite{TrummerKoch2015}, who investigate scaling behavior and show a speed-up of several orders of magnitude over classical optimization algorithms, and \cite{Cao2016}, who in contrast do not detect any quantum speedup for the \textit{set cover with pairs} problem, one of Karp's $21$ NP-complete problems. Other work include graph partitioning via quantum annealing on DW in the context of QMD (quantum molecular dynamics) applications \citep{Mniszewski2016,Ushijima2017}, in which graph partitioning is shown to reduce the computational complexity of QMD. Test sets for integer optimization are investigated in \cite{Coffrin2017}, who observe an advantage of DW over Gurobi~\citep{gurobi} both in terms of speed and quality of solution.

In \cite{Stollenwerk2015}, a general introduction to the DW architecture and the representation of problem instances in Ising and QUBO format is given as well as a QUBO formulation for the maximum clique problem. However, the authors do not actually report any computation results for finding maximum cliques on DW, nor do they compare DW to state-of-the-art heuristic solvers. In contrast to \cite{Stollenwerk2015}, we solve the maximum clique problem on DW for a variety of test graphs and compare its solutions to the ones of state-of-the-art classical solvers. Moreover, we present a graph splitting algorithm allowing to solve problem sizes larger than those embeddable on DW, analyze its scaling behavior, and investigate the influence of alternative QUBO formulations on the solution quality.

In \cite{Boothby2016}, the authors consider finding large clique minors in the DW hardware \textit{Chimera} graph $C(m,n,l)$, defined as the $m \times n$ grid of $K_{l,l}$ complete bipartite graphs (also called \textit{unit cells}). The authors present a polynomial time algorithm for finding clique minors in the special case of the Chimera graph only. Such clique minors are needed to embed problem instances of arbitrary connectivity onto the current and future Chimera architectures, given the problem size is not larger than the clique minor. In contrast, in the present article we consider finding maximum cliques in arbitrary graphs with DW by minimizing a QUBO for the maximum clique problem (applied to the user-specified arbitrary input graph). This step requires an embedding of our QUBO onto DW's Chimera graph, for which the algorithm of \cite{Boothby2016} can be beneficial. However, the work of \cite{Boothby2016} does not substitute for the embedding and minimization of a QUBO when finding cliques in arbitrary graphs as considered in our work.

Instead of attempting a full solution via DW, other publications propose using a quantum annealer to assist in finding a solution of certain problem classes, which often are of a practical and thus more complex nature. For instance, \cite{DridiAlghassi2016} consider computing the (algebraic) homology of a data point cloud and propose to reduce this computation to a minimum clique covering, which is then suggested to be solved using DW. No empirical results are presented. A real-world application (the network scheduling problem) is considered in \cite{Wang2016}. The authors demonstrate an advantage of quantum over simulated annealing; moreover, they show how to obtain more admissible solutions with DW by introducing an additional weight into the QUBO that incrases the gap between linear and quadratic QUBO terms. In \cite{NguyenKenyon2017}, a sparse coding model is trained using samples obtained via DW from a Hamiltonian with $L_p$ sparseness penalty. A graph flow problem in real-world traffic network analysis is considered in \cite{Neukart2017}, who employ gps coordinates of cars in Beijing as training data.

Another class of publications is concerned with the theory of quantum annealing and the problem of benchmarking quantum computations. For instance, \cite{RogersSingleton2016} empirically verify the known phase transition in magnetization for the 2D Ising model with DW. In \cite{Thulasidasan2016}, the author proposes to run Markov Chain Monte Carlo using samples generated by DW from a suitable Boltzmann distribution. The question whether random spin-glass problems are a suitable type of problem to detect a quantum advantage over classical approaches is considered in \cite{PerdomoOrtiz2017}, who also study the problem of benchmarking quantum annealing vs.\ classical CMOS computation.

\section{Solving MC on D-Wave}
\label{section_methods}
This section introduces the DW chip architecture and briefly presents three tools provided by D-Wave Inc.\ to submit quadratic programs to the quantum computer. We also introduce the QUBO formulation of MC needed to submit an MC instance to DW. The section concludes with an algorithmic framework designed to solve instances of MC which are not embeddable on DW.

\subsection{DW hardware and software}
\subsubsection{The Qubit architecture}
\label{section_architecture}
DW operates on roughly $1000$ qubits. The precise number of available qubits varies from machine to machine (even of the same type) due to manufacturing errors which render some of the qubits inoperative. The qubits are connected using a specific type of network called  \textit{Chimera} graph, $C_{12,12,4}$ (see Fig.~\ref{fig:chimera}), comprised of a lattice of $12\times 12$ cells, where each cell is a $4\times 4$ complete bipartite graph. DW can naturally solve Ising and QUBO problems where non-zero quadratic terms are represented by an edge in the Chimera graph.

\begin{figure}
  \centering
  \includegraphics[width=0.5\textwidth]{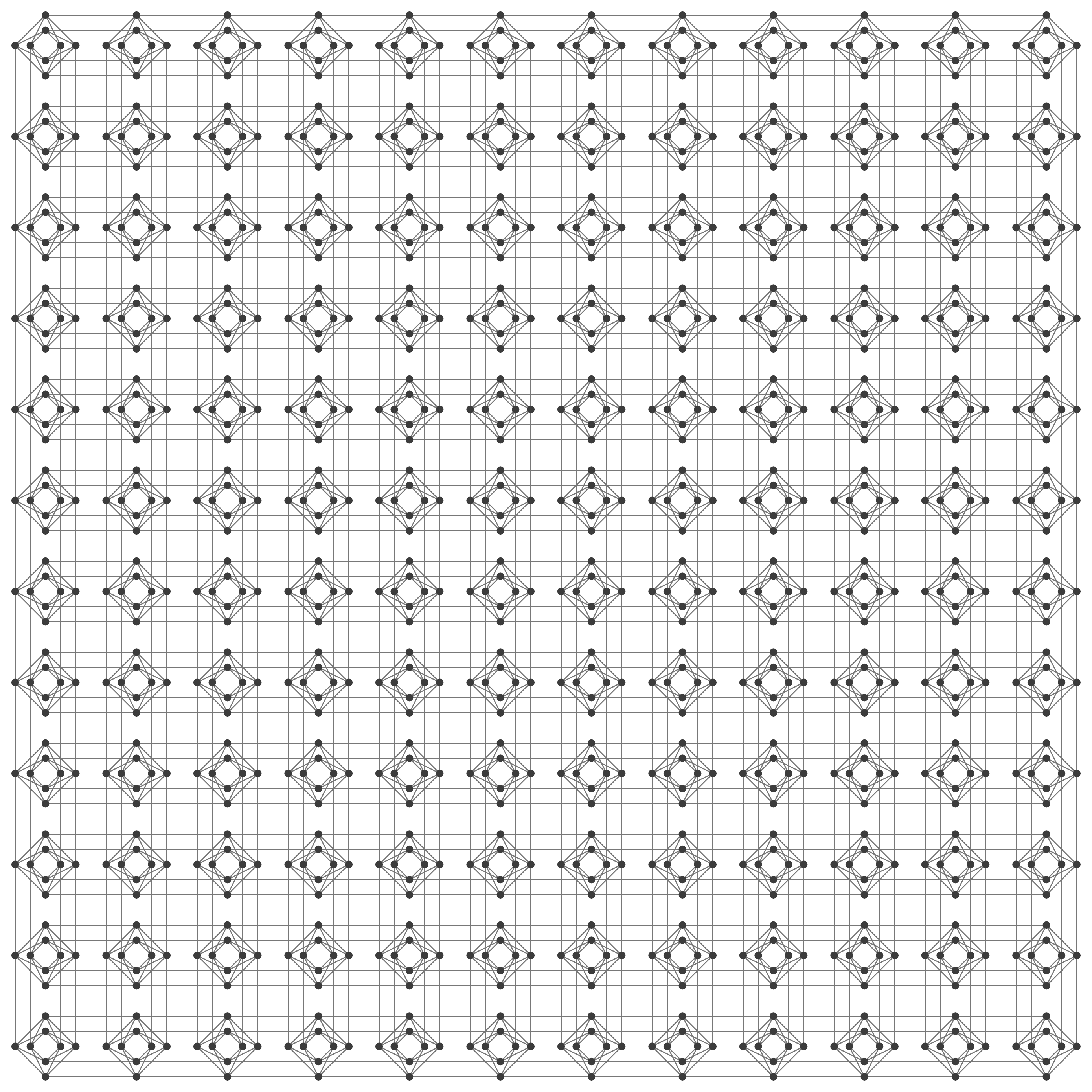}
  \caption{The Chimera $C_{12,12,4}$ graph of 1152 vertices (qubits) and 3360 edges (couplers). LANL's D-Wave 2X chip has usable only 1095  qubits and 3061 couplers due to manufacturing defects.}
  \label{fig:chimera}
\end{figure}

The particular architecture of the qubits implies two important consequences: First, the chip design actually only allows for direct pairwise interactions between two qubits which are  physically adjacent on the chip. For pairwise interactions between qubits not physically connected, a \emph{minor embedding} of the graph describing the non-zero structure of the Hamiltonian matrix into the Chimera type graph is needed,  which maps a logical variable into one or several physical qubits on the chip. Minor embeddings are hence necessary to ensure arbitrary connectivity of the logical variables in the QUBO. The largest complete graph that the DW can embed in theory has $1+4 \cdot 12 = 49$ vertices. In practice, the largest embeddable graph is slightly smaller ($n\approx 45$) due to missing qubits arising in the manufacturing stage.

When more than one qubit is used to represent a variable, that set of qubits is called a \textit{chain}. The existence of chains has two vital consequences, which will play an important role in the analyses of Section~\ref{sect:experiments}. On the one hand, the need for chains uses up qubits, which would otherwise be available to represent more variables in the quadratic program, thereby reducing the maximum problem sizes that can directly be solved on DW. This is the reason for the relatively small sizes of $N=45$ for  QUBO problems \eqref{eq:hamilt} that fit onto DW when the corresponding Hamiltonians are dense (contain nearly all quadratic terms), despite the fact that more than 1000 qubits are available in DW.

On the other hand, due to the imperfections of the quantum annealing process caused by environmental noise, limited precision, and other shortcomings, solutions returned by D-Wave do not always correspond to the minimum energy configuration. In the case of chains, all qubits in a chain encode the same variable in \eqref{eq:hamilt} and hence should have the same value, but for the reasons outlined above this may not be the case. This phenomenon is called a \textit{broken chain}, and it is not clear which value should be assigned to a variable if its chain is broken. Clearly, chains can be ensured to not break by increasing their  coupler weights, but as we will see in the next section this may significantly reduce the accuracy of the solver.

\subsubsection{D-Wave solvers}
\label{section_dwave_solvers}
D-Wave Inc.\ provides several tools that help users submit their QUBO problems to the quantum processor, perform the annealing, apply necessary pre- and post-processing steps, and format the output. This section briefly describes several such tools used in this article.

\paragraph{Sapi}
Sapi stands for Solver API and provides the highest level of control one can have over the quantum annealer. It allows the user to compute minor embeddings for a given Ising or QUBO problem, to choose the number of annealing cycles, or to specify the type of post-processing. Sapi interfaces for the programming languages \textit{C} and \textit{Python} are available. One can also use a pre-computed embedding of a complete $45$-vertex graph, thus avoiding the need to run the slow embedding algorithm.

\paragraph{QBsolv}
QBsolv is a tool that can solve problems in QUBO format which are of a size that cannot natively fit onto DW. Larger problems (with more variables or more connections than can be mapped onto the corresponding Chimera graph) are analyzed by a hybrid algorithm, which identifies a small number of significant rows and columns of the Hamiltonian. It then defines a  QUBO on that subset of variables which fits DW, solves it, and extends the found solution to a solution of the original problem.

\paragraph{QSage}
In contrast to Sapi or QBsolv, QSage is a blackbox hybrid solver which does not require a QUBO or Ising formulation as input. Instead, QSage is able to minimize any function operating on a binary input string of arbitrary size. For this it uses a tabu search algorithm enhanced with DW-generated low-energy samples near the current local minimum. To ensure that also input sizes larger than the DW architecture can be processed, QSage optimizes over random substrings of the input bits.

\subsection{QUBO formulations of MC}
\label{sec:maxc}
Recall that a QUBO problem can be written as 
\begin{equation}\label{eq:qubo}
\begin{aligned}
& \underset{x_i \in \{0,1\} }{\text{minimize}}
& \!\!\!\!&H=\sum_{1\leq i<j\leq N} a_{ij} x_i x_j,
\end{aligned}
\end{equation}
where the weights $a_{ij}$, $i \neq j$, are the quadratic terms and $a_{ii}$ are the linear terms (since $x_i^2=x_i$ for $x_i\in\{0,1\}$).

There are multiple ways to formulate the MC problem as a QUBO. One of the simplest is based on the equivalence between MC and the maximum independent set problem. An independent set $S$ of a graph $H$ is a set of vertices with the property that for any two vertices $v,w\in S$, $v$ and $w$ are not connected by an edge in $H$. It is easy to see that an independent set of $H = (V, \overline{E})$ defines a clique in graph $G=(V,E)$, where $\overline{E}$ is the complement of set $E$. Therefore, looking for the maximum clique in $G$ is equivalent to finding the maximum independent set in $H$. The corresponding constraint formulation for MC is
\begin{equation}\label{eq:mis}
\begin{aligned}
& \underset{x_i \in \{0,1\} }{\text{maximize}}
& \!\!\!\!&\sum_{i=1}^Nx_i\\
& \text{subject to}
& & \sum_{(i,j)\in \overline{E}} x_ix_j=0,
\end{aligned}
\end{equation}
where $G=(V,E)$ is the input graph and $\overline{E}$ is the complement of $E$. The equivalent unconstrained (QUBO) minimization of \eqref{eq:mis}, written in the form \eqref{eq:qubo}, is
\begin{equation}
\label{eq:mis2}
H=-A\sum_{i=1}^Nx_i+B\sum_{(i,j)\in \overline{E}} x_ix_j,
\end{equation}
where one can determine that the coefficients/penalties $A$ and $B$ can be chosen as $A=1$, $B=2$ (see \cite{Lucas2014}). A disadvantage of the formulation \eqref{eq:mis2} is that $H$ contains an order of $N^2$ quadratic terms even for sparse graphs $G$, which limits the size problems for which MC can be directly solved on DW.

\subsection{Solving larger MC instances}
\label{section_large_instances}
To solve the MC problem on an arbitrary graph, we develop several algorithms that reduce the size of the input graph by removing vertices and edges that do not belong to a maximum clique and/or split the input graph into smaller subgraphs of at most 45 vertices, the maximal size of a complete graph embeddable on DW. Let $G(V,E)$ be a connected graph of $n$ vertices.

\subsubsection{Extracting the $k$-core}
The $k$-core of a graph $G=(V,E)$ is the maximal subgraph of $G$ whose vertices have degrees at least $k$. It is easy to see that if $G$ has a clique $C$ of size $k+1$, then $C$ is also a clique of the $k$-core of $G$ (since all vertices in a $k$-clique have degrees $k-1$). Therefore, finding a maximum clique of size no more than $k+1$ in the original graph $G$ is equivalent to finding such a clique in the $k$-core of $G$ (which might be a graph of much smaller size).

One can compute the $k$-core iteratively by picking a vertex $v$ of degree less than $k$, removing $v$ and its adjacent edges,  updating the degrees of the remaining vertices, and repeating while such a vertex $v$ exists. The algorithm can be implemented in optimal $O(|E|)$ time~\citep{DBLP:journals/corr/cs-DS-0310049}. 

We also apply another reduction approach, which we refer to as edge $k$-core, to reduce the size of an input graph $g$ using a known lower bound \textit{lower\_bound} on the clique size. This approach combining $k$-core and edge $k$-core is given in pseudo-code notation as Algorithm~\ref{alg:Graph-reducing-algorithm}.

\begin{algorithm}
  \caption{\texttt{Graph reducing $k$-core based algorithm}\label{alg:Graph-reducing-algorithm}}
  \SetKwInOut{Input}{input}
  \SetKwInOut{Output}{output}
  \SetKwProg{proc}{def}{}{}
  \proc{\textnormal{reduce\_graph(Graph g, int lower\_bound):}}{
    extract\_k\_core(g, lower\_bound)\\
    Vertex v = choose\_random\_vertex(g)\\
    \For{\textnormal{\textbf{each} vertex n in neighbors(g, v)}}{
      Set nv = neighbors(g,v)\\
      Set nn = neighbors(g,n)\\
      List common\_neighbors = intersection(nv, nn)\\
      \If{\textnormal{length(common\_neighbors) $<$ lower\_bound-2}}{
	remove\_edge(g, v, n)
      }
    }
    extract\_k\_core(g, lower\_bound)
  }
\end{algorithm}

In Algorithm~\ref{alg:Graph-reducing-algorithm}, we first aim to reduce the size of $g$ by simply extracting its $k$-core, where $k$ is set to the currently known lower bound. It is easily shown that for two vertices $v,w$ in a clique of size $c$, the intersection of the two neighbor lists of $v$ and $w$ has size at least $c-1$. We therefore choose a random vertex $v$ in $sg$ and remove all edges $(v,e)$ satisfying $|N(v) \cap N(e)|<lower\_bound-2$ (here $N(v)$ denotes the set of neighbor vertices of $v$), as such edges cannot be part of a clique with size larger than $lower\_bound$. Since this changes the graph structure, we attempt to extract the $lower\_bound$-core at the end again before returning the reduced graph.

\subsubsection{Graph partitioning}
\label{sec:graphpartitioning}
This divide-and-conquer approach aims at dividing $G$ into smaller subgraphs, solves the MC problem in each of these subgraphs, and combines the subproblem solutions into a solution of the original problem. If one uses standard (edge-cut) graph partitioning, which divides the vertices of the graph into a number of roughly equal parts so that the number of \textit{cut edges}, or edges with endpoints in different parts, is minimized, then the third step, combining the subgraph solutions, will be computationally very expensive. Instead, we will use CH-partitioning, recently introduced in \cite{Djidjev2016}. 

In \textit{CH-partitioning}, there are two levels of dividing the vertices of $G$ into subsets. In the \textit{core partitioning}, the set $V$ of vertices is divided into nonempty \textit{core} sets $C_1,\dots,C_s$ such that $\bigcup_iC_i=V$ and $C_i\cap C_j=\emptyset$ for $i\neq j$. There is one \textit{halo} set $H_i$ of vertices for each core set $C_i$, defined as the set of neighbor vertices of $C_i$ that are not from $C_i$. Recall that a vertex $w$ is a neighbor of a vertex $v$ iff there is an edge between $v$ and $w$. We define the cost of the CH-partitioning ${\mathcal P}=(\{C_i\},\{H_i\})$ as 
\begin{equation}
\label{eq:CHcost}
\mbox{cost}({\mathcal P})=\max_{1\leq i\leq s}(|C_i|+|H_i|).
\end{equation}
The \textit{CH-partitioning problem} is finding a CH-partitioning of $G$ of minimum cost. The next statement shows how CH-partitions can be used for solving MC in larger graphs.

\begin{proposition}\label{prop:CHpart}
  Given a CH-partitioning $(\{C_i\},\{H_i\})$ of a graph $G$, the size of the maximum clique of $G$ is equal to $\max_i\{k_i\}$, where $k_i$ is the size of a maximum clique of the subgraph of $G$ induced by $C_i\cup H_i$.
\end{proposition}

\begin{proof}
Let $K$ be a maximum clique of $G$ and let $v$ be any vertex of $K$. Since, by definition of CH-partitioning, $\bigcup_iC_i=V$, where $V$ is the set of the vertices of $G$, then $v$ belongs to some core $C_j$. 

We will next show that for any vertex $w\neq v$ from $K$, $w\in C_j\cup H_j$, which will imply that all vertices of $K$ are in $C_j\cup H_j$, implying the correctness of the proposition. 

If $w\in C_j$ then the claim follows. 

Assume that $w\not\in C_j$. We will show that $w\in H_j$. Since $K$ is a clique, there is an edge between any two vertices from it, and hence there is an edge between $v$ and $w$. Since, by definition, $H_j$ consists of all neighbors of vertices from $C_j$ that are not in $C_j$, $v\in C_j$, $w\not\in C_j$,  and $w$ is a neighbor of $v$, then $w\in H_j$.
\end{proof}

Using Proposition~\ref{prop:CHpart}, the solutions to all subproblems of a CH-partitioning can be combined into a solution of the original problem at an additional cost of only $O(s)=O(n)$, where $s$ is the number of the sets of the partition.

One may conjecture that increasing $s$ in \eqref{eq:CHcost} will always reduce the cost, but this is not always the case. If the minimum cost is achieved for $s=1$, or if some of the parts of the partition are still too large, then the method in the next subsection might be applied.

\subsubsection{Vertex splitting}
\label{sec:vertexsplitting}
This method is similar to a special case of the previous one, obtained by choosing $s=2$, letting $C_1$ contain only a single vertex $v$, and letting $C_2$ contain all other vertices $V\setminus \{v\}$. Moreover, while the halo $H_1$ of $C_1$ is defined as above, we set $C_1=\emptyset$ and $H_2=\emptyset$. As a result, $G$ is divided into two subgraphs, $G_1$ containing all neighbors of $v$ without $v$ itself, and $G_2$ containing all vertices of $G$ except $v$, see Fig.~\ref{fig:vertex_splitting}. Because this partitioning is uniquely determined by a single vertex, we call it a \textit{vertex-splitting partitioning}. The cost of such a partitioning is again given by \eqref{eq:CHcost}.

\begin{figure}
  \centering
  \includegraphics[width=0.5\textwidth]{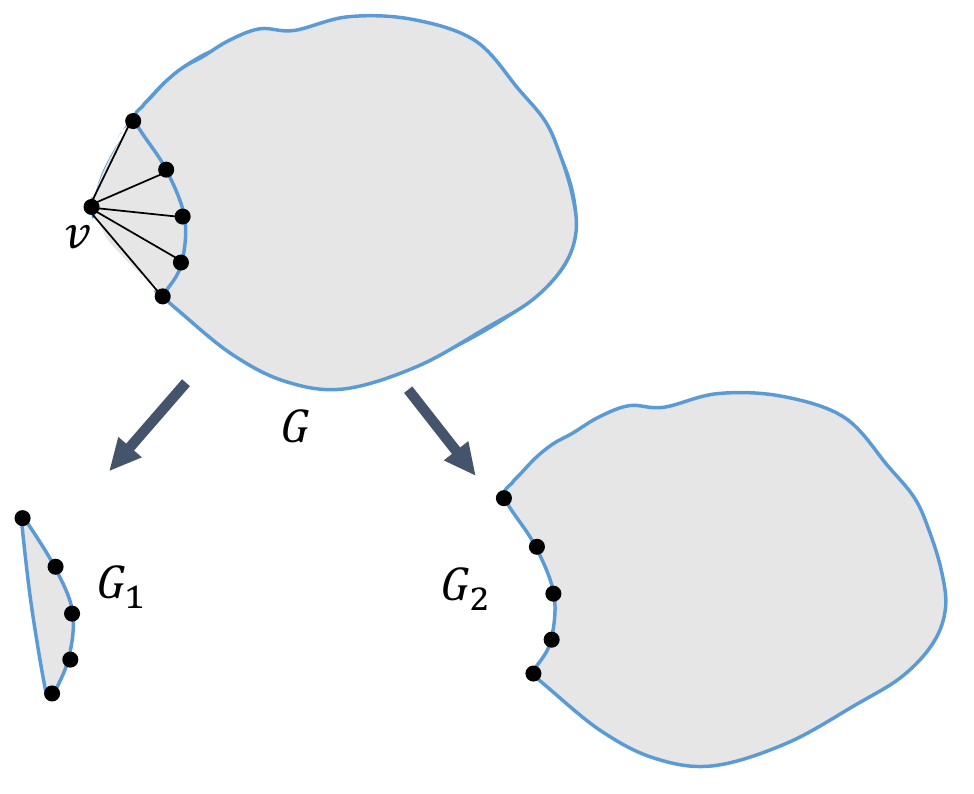}
  \caption{Illustration of the vertex splitting algorithm.}
  \label{fig:vertex_splitting}
\end{figure}

\begin{proposition}
\label{prop:VSpart}
  Given a vertex-splitting partitioning of $G$, $(\{C_1,C_2\},\{H_1,H_2=\emptyset\})$, the size of the maximum clique of $G$ is equal to $\max\{k_1+1,k_2\}$, where $k_i$, $i=1,2$,  is the size of a maximum clique of the subgraph of $G$ induced by $C_i\cup H_i$.
\end{proposition}

\begin{proof}
Let $v$ be the vertex that defines the partition. If there is a maximum clique of $G$ that contains $v$ let $K$ be such a clique, otherwise let $K$ be any maximum cliques of $G$. Consider the following two cases.

Case 1: $v$ belongs to $K$. Then the set of the vertices of $K$ consists of $v$ and a subset $V_1$ of vertices from $G_1$. Moreover, since there is an edge between any two vertices of $K$, there is an edge between any two vertices of $V_1$, which means that $V_1$ defines a clique $K_1$ in $G_1$. Assume that $K_1$ is not a maximum clique of $G_1$, i.e., there exists a clique $K_1'$ in $G_1$ with more vertices than $K$. Then adding $v$ to the vertices of $K_1'$ will result in a clique in $G$ of size larger than $K$, which is a contradiction to the choice of $K$. Hence $K_1$ is a maximum clique in $G_1$, whose size was denoted by $k_1$. Since $K$ consists of $v$ and the vertices of $K_1$, its size is $k_1+1$. Moreover, $G_2$ cannot have a clique larger than $K$ since any clique in $G_2$ is also a clique in $G$. Hence, $k_2\leq |K|=k_1+1$ and $|K|=\max\{k_1+1,k_2\}$.

Case 2: $v$ does not belong to $K$. Then, $K$ is entirely contained in $G_2$ and hence $|K|\leq k_2$. On the other hand, $G_2$ cannot have a larger clique than $|K|$ since any clique in $G_2$ is also a clique in $G$, hence $|K|=k_2$. Moreover, by the choice of $K$, any clique containing $v$ is of size less than $K$, so $|K|>k_1+1$, and therefore $|K|=\max\{k_1+1,k_2\}$.
\end{proof}

Since $H_2=\emptyset$, vertex splitting can be used in cases where CH-partitioning fails. Moreover, if there is a vertex of degree less than $n-1$, this method will always create subproblems of  size smaller than the original one. However, the total number of subproblems resulting from the repeated use of this method can be too large. A more efficient algorithm can be obtained if all the above methods are combined.

\subsubsection{Combining the three methods}
We use the following algorithm to decompose a given input graph $G$ into smaller MC instances fitting the DW size limit. We assume that the size $k+1$ of the maximum clique is known. (Otherwise, use the procedure of this section in a binary-tree search fashion to determine the size of the maximum clique. This increases the running time by a factor $O(\log k)=O(\log n)$ only.) We also have an implementation that, instead of ``guessing'' the \textit{exact} value of $k$, uses lower bounds on $k$ determined by the size of the largest clique found so far. 

The algorithm works in two phases. First, we apply the $k$-core algorithm on the input graph and then CH-partitioning on the resulting $k$-core. Consequently, we keep a list $L$ of subgraphs (ordered by their number of vertices), which is initialized with the output of the CH-partitioning step. In each iteration and until all produced subgraphs fit the (DW) size limit, we choose a vertex $v$ from the largest subgraph $sg$, extract the subgraph $ssg$ induced by $v$ and its neighbors and remove $v$ from $sg$. The $k$-cores of the two subgraphs produced at this iteration are then inserted into $L$. Second, we compute the maximum clique on DW for any subgraph in $L$ of size small enough.

Algorithm~\ref{alg:Graph-splitting-algorithm} gives the pseudo-code of this approach. It returns a list of subgraphs of an input graph $g$ sorted in increasing order of their number of vertices as well as an updated lower bound on the maximum clique size. The parameters of Algorithm~\ref{alg:Graph-splitting-algorithm} are the input graph $g$, a maximal number of vertices \textit{vertex\_limit} for which the maximum clique problem is solved directly on a subgraph, and a lower bound on the clique size found so far (\textit{lower\_bound}). All returned subgraphs have the property that their size is at most \textit{vertex\_limit}. Since the algorithm attemps to solve MC exactly on graphs not larger than \textit{vertex\_limit}, the parameter \textit{vertex\_limit} in our case can be set to the maximal number of vertices embeddable on DW.

\begin{algorithm}
  \caption{\texttt{Graph splitting algorithm}\label{alg:Graph-splitting-algorithm}}
  \SetKwInOut{Input}{input}
  \SetKwInOut{Output}{output}
  \SetKwProg{proc}{def}{}{}
  \proc{\textnormal{split(Graph g, int vertex\_limit, int lower\_bound):}}{
    List subgraphs = [g]\\
    \While{\textnormal{length(subgraphs[-1]) $>$ vertex\_limit}}{
      Graph sg = subgraphs.pop()\\
      Vertex v = choose\_vertex(sg)\label{alg_line:remove_v}\\
      Graph ssg = extract\_subgraph(v, sg)\\
      remove\_vertex(v, sg)\\
      reduce\_graph(sg, lower\_bound)\\
      \If{\textnormal{length(sg) $>$ 0}}{
	\If{\textnormal{length(sg) $<$= vertex\_limit}}{
	  lower\_bound = solve(sg)
	}
	\Else{
	  sorted\_insert(subgraphs, sg)
	}
      }
      reduce\_graph(ssg, lower\_bound)\\
      \If{\textnormal{length(ssg) $>$ 0}}{
	\If{\textnormal{length(ssg) $<=$ vertex\_limit}}{
	  lower\_bound = solve(ssg)
	}
      }
    }
    return subgraphs, lower\_bound
  }
\end{algorithm}

Algorithm~\ref{alg:Graph-splitting-algorithm} works as follows. First, a sorted list of graphs called \textit{subgraphs} (sorted in descending order of the degree of the subgraphs) is created and initialized with $g$. As long as the largest subgraph (denoted as \textit{subgraphs}[-1] in \textit{Python} notation) has at least \textit{vertex\_limit} nodes, the current largest graph \textit{sg} in the list (command \textit{pop()}) is returned, \textit{sg} is removed from list \textit{subgraphs}, and a vertex $v$ is chosen according to some rule specified in the function \textit{choose\_vertex} (see the end of Section~\ref{section_large_instances} for possible approaches). Then, the induced subgraph \textit{ssg} by vertex $v$ is extracted and deleted from \textit{sg}.

A graph reduction step via the function \textit{reduce\_graph} is then applied to \textit{sg} which reduces the size of the graph using the currently known best lower bound \textit{lower\_bound} on the clique size. The graph reduction is given separately as Algorithm~\ref{alg:Graph-reducing-algorithm}.

Suppose \textit{sg} still contains vertices after reduction. If the degree of \textit{sg} after reduction is less than \textit{vertex\_limit}, we attempt to solve the MC problem exactly on \textit{sg} using some function \textit{solve()} (for instance via DW) and update \textit{lower\_bound}. Otherwise, \textit{sg} is inserted again into the list \textit{subgraphs}.

The same step is repeated for the subgraph \textit{ssg} with the exception that \textit{ssg} does not have to be re-inserted into list \textit{subgraphs} at the end. This is because the subgraph induced by a single vertex $v$ either contains a clique or can be removed.

Removing a vertex $v$ in line \ref{alg_line:remove_v} of Algorithm~\ref{alg:Graph-splitting-algorithm} decreases the size of \textit{sg} by one in each iteration, thus the algorithm terminates in finite time once all generated subgraphs have size at most \textit{vertex\_limit}. 

Lastly, we describe our procedure \textit{choose\_vertex(sg)} for choosing the next vertex to be removed from $sg$. A vertex with high degree will potentially greatly reduce the size of $sg$, however at the expense of also producing a large subgraph $ssg$. In order to maximize the impact of removing a vertex, we successively try out three choices: a vertex of highest degree, a vertex of median degree and, if necessary, a vertex of lowest degree in $sg$. If the vertex of lowest degree has degree $|V|-1$, then $sg$ is a clique: In this case, solving MC on $sg$ can be omitted and $lower\_bound$ can be updated immediately.

\section{Experimental analysis}
\label{sect:experiments}
The aim of this section is to investigate if a quantum advantage for the MC problem can be detected for certain classes of input graphs. To this end, we compare the DW solvers of Section~\ref{section_dwave_solvers} to classical ones on various graph instances -- from random small graphs that fit the DW chip to (larger) graphs tailored to perfectly fit DW's Chimera architecture. We also evaluate our graph splitting routine of Section~\ref{section_large_instances} on large MC instances. First we briefly describe classical solvers that will be used in the comparison.

\subsection{Classical solvers}
\label{sec:solvers}
Apart from the tools provided by D-Wave Inc., we employ classical solvers in our comparison, consisting of: A simulated annealing algorithm working on the Ising problem (SA-Ising), a simulated annealing algorithm specifically designed to solve the clique problem (\textit{SA-clique}, see \cite{geng2007simple}), softwares designed to find cliques in heuristic or exact mode (the \textit{Fast Max-Clique Finder} fmc, see \cite{pattabiraman2013fast}), the software tool \textit{pmc} (see \cite{rossi2013fast}), and the Gurobi solver \citep{gurobi}.

\paragraph{SA-Ising}
This is a simulated annealing algorithm working on an Ising problem formulation. The initial solution is a random solution, and a single move in the simulated algorithm is the flip of one random bit.

\paragraph{SA-clique}
We implemented a simulated annealing algorithm specifically designed to find cliques, as described in \cite{geng2007simple}. As SA-clique only finds cliques of a user-given size $m$, we need to apply a binary search on top of it to find the maximum clique size. Its main parameter is a value $\alpha$ controlling the geometric temperature update of the annealing in each step (that is, $T_{n+1} = \alpha T_n$). A default choice is $\alpha=0.9996$. A value closer to $1$ will yield a better solution but will increase the computation time.

\paragraph{Fast Max-Clique Finder (fmc, pmc)}
These two algorithms are designed to efficiently find a maximum clique for a large sparse graph. They provide  exact and heuristic search modes. We use version 1.1 of software \textit{fmc} \citep{pattabiraman2013fast} and \textit{pmc} (github commit 751e095) \citep{rossi2013fast}.

\paragraph{Post-processing heuristics alone (PPHa)}
The DW pipeline includes a post-processing step: First, if chains exist, a majority vote is applied to fix any broken chains. Then a local search is performed to ensure that any solution is indeed a local minimum (the raw solutions coming from DW might not be in a local minimum, see \cite{dwavepostprocess24}). For a given solution coming out of the pipeline, one might wonder what the relative contributions of DW and of the post-processing step are. For some small and simple problems, the post processing step \emph{alone} might be able to find a good solution.

We try to answer this issue by solely applying the post-processing step, and by comparing the result with the one obtained by quantum annealing. However, post-processing by DW runs on the DW server and is not available separately. 

To enable us to still use the DW post-processing alone, we employ the following procedure. We set a very high absolute chain strength (e.g., $1000$ times greater than the largest weight in our Ising problem), and turn on the \textit{auto-scale} feature mapping QUBO weights to the interval $[-1,1]$. Because of the limited precision of the DW hardware (DW maps all QUBO weigths to $16$ discrete values within $[-1,1]$), chain weights will be set to the minimum value $-1$ while all other weights will be scaled down to $0$. In this way, the quantum annealer will only satisfy the chains rather than the actual QUBO we are interested in. As chains will not be connected to other chains, and as all linear terms will be zero, each chain will be assigned a random value $-1$ or $+1$. Applying the DW post-processing step to such a QUBO with large chain weights will therefore result in the post-processing step being called with a random initial solution. We hence expect to obtain results stemming from the post-processing step only (with random starting point). This method will be referred to as PPHa, \emph{post-processing heuristic alone}.

\paragraph{Gurobi}
\textit{Gurobi} \citep{gurobi} is a mathematical programming solver for linear programs, mixed-integer linear and quadratic programs, as well as certain quadratic programs. We employ \textit{Gurobi} to solve given QUBO problems (Ising problems can be solved as well, nevertheless \textit{Gurobi} explicitly allows to restrict the range of variables to binary inputs, making it particularly suitable for QUBO instances). Instead of solving MC directly with Gurobi, we solve the dual problem, that is we computed a maximum independent set on the complement graph.

\subsection{Small graphs with no special structure}
\label{sec:smallgraphs}
\begin{figure}
  \centering
  \includegraphics[width=0.5\textwidth]{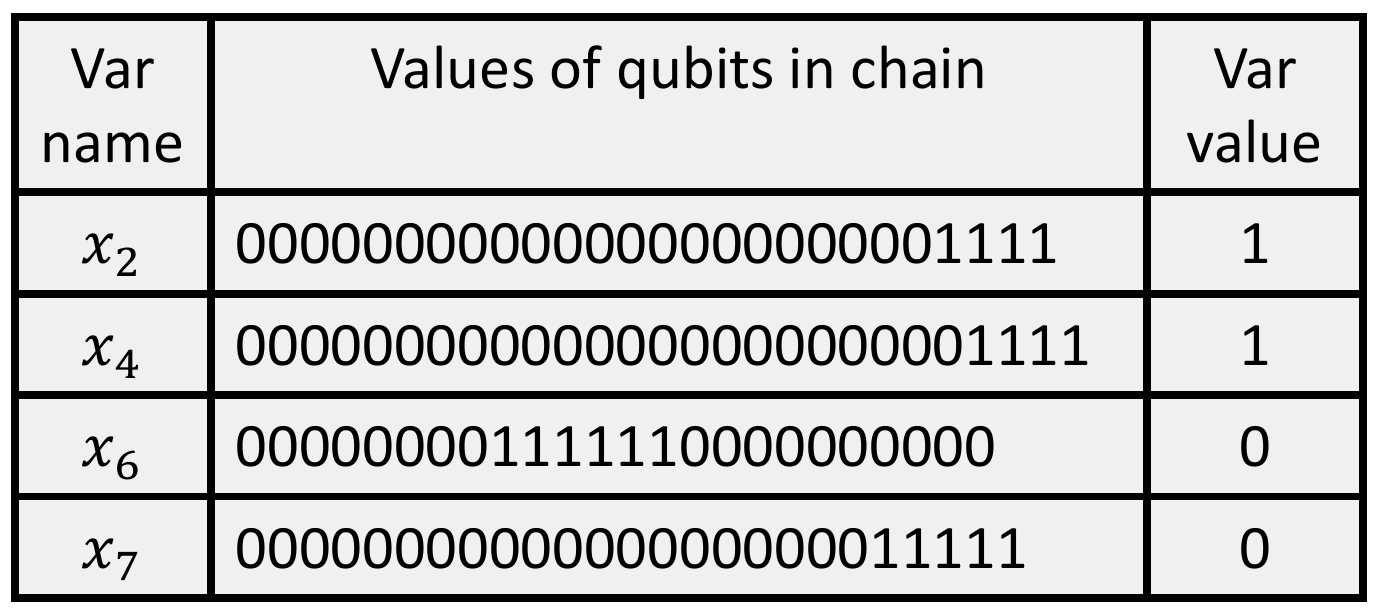}
  \caption{The first four broken chains (out of 16) produced by DW on a test 45-vertex graph. The first column shows the name of the variable the chain corresponds to and the third column gives the correct value for that variable.}
  \label{fig:chains}
\end{figure}

\begin{table*}
  \centering
  \begin{tabular}{|l||c||l|l|l|l|l|l|l|}
  \hline
  Graph & Max.~clique size & \multicolumn{7}{c|}{Runtime [s]}\\
  && Sapi & PPHa & QBsolv & \textit{fmc} & \textit{pmc} &SA & Gurobi\\
  \hline
  p=0.3 & 5	& 0.15 & 0.15 & 0.05 & $8\cdot 10^{-6}$ & $3\cdot 10^{-5}$ & 0.15& 102\\
  p=0.5 & 8	& 0.15 & 0.15 & 0.06 & $3\cdot 10^{-4}$ & $5\cdot 10^{-5}$ & 0.37& 38\\
  p=0.7 & 13	& 0.15 & 0.15 & 0.04 & 0.002 & $8\cdot 10^{-5}$ & 0.19& 33\\
  p=0.9 & 20	& 0.15 & 0.15 & 0.04 & 0.135 & $8\cdot 10^{-5}$ & 0.28& 2\\
  \hline
  \end{tabular}
  \caption{Running time on 45 vertex random graphs. The edge probability used to generate those graphs is given in the first column. Since for such small graphs, every software returned the correct solution, we only report the running times. Gurobi solves the dual problem, leading to reversed graph densities and timings.}
  \label{tab:45g}
\end{table*}

We generate four random graphs with increasing edge densities for our experiments. We considered edge probabilities ranging from $0.3$ to $0.9$ in steps of $0.05$. We compare the execution times of DW using the Sapi interface and the different solvers listed in Section~\ref{section_dwave_solvers} to the classical solvers of Section~\ref{sec:solvers}.

Results are shown in Table~\ref{tab:45g}. For small graphs, every solver returns a maximum clique, therefore the table shows execution times only. We can see that (a) software solvers are much faster than DW, with \textit{pmc} being the fastest by several order of magnitudes; (b) DW and \textit{PPHa} exhibit equal results and execution times. This shows that for these small graphs, even the simple software heuristic included in the DW pipeline is capable of solving the MC problem. The similar performance of DW and PPHa therefore makes it impossible to distinguish between the contributions from the post-processing heuristic and the actual quantum annealer; (c) Gurobi finds the best solution as well (for the dual of MC, the maximum independent set problem, thus timings decrease in the last column of Table~\ref{tab:45g}), but since Gurobi is an exact solver, its running time is higher than the one of the other methods. We note that the timings for Gurobi are for finding the best solution -- letting Gurobi run further to subsequently prove that a found solution is optimal requires a far longer runtime.

Moreover, we have observed in eq.~\eqref{eq:mis2} in Section~\ref{sec:maxc} that the QUBO for MC leads to an order of $N^2$ quadratic terms even for sparse graphs. This in turn typically causes the QUBO matrix to be very dense, making it difficult to embed the QUBO onto the Chimera graph. If embedding the QUBO is indeed possible, then usually at the cost of incurring long chains. This is due to the fact that a dense QUBO necessarily contains a large number of couplers between qubits not adjacent on the DW chip, thus requiring re-routing through chain qubits. In our experiments we observe that this is a delicate case for the quantum annealer: Using high coupler strengths for the chains results in consistent chains after annealing, but comes at the cost of downscaling the actual QUBO weights, thus leading to meaningless solutions. Lower chain strengths often cause many of the chains to be broken, i.e.\ the physical qubits constituting the chains have different values. Therefore some processing needs to be applied to obtain valid solutions. The most simple one is a majority vote, however all postprocessing rules offered by DW are merely heuristic ways of assigning final values to the qubits. It is not guaranteed that a weighting scheme exists which preserves the QUBO and prevents chains from breaking at the same time.

As an example, Fig.~\ref{fig:chains} shows the first four broken chains in a typical DW execution of the MC problem on a $45$ vertex graph. The chain for $x_2$ has more zeros and less ones than the one for $x_7$, yet after the DW postprocessing algorithm was applied, the variables got correct values $x_2=1, x_7=0$ (with apparently PPHa overwriting the inferior DW solution). Our experiments with randomly assigned values to broken chains (see the discussion for PPHa in Section~\ref{sec:solvers}) similarly show that  accurate solutions  obtained for small graphs are often mostly due to the post-processing algorithm rather than the quantum annealing by DW.

\subsection{Graphs of sizes that fit DW}
\label{sec:artig}
In Section \ref{sec:smallgraphs}, we performed experiments with random graphs that can be embedded onto DW. We observed that highly optimized software solvers outperformed DW in terms of speed. This is due to the fact that the largest random graphs we are sure to embed on DW (around 45 vertices) are still comparably small and can hence be solved efficiently with an optimized heuristic. In order to detect a difference between DW and classical solvers, we need to consider larger graphs. In this section we will analyze the behavior of the quantum annealer on subgraphs of DW's chimera graph, i.e.\ the largest graph we can embed on the DW architecture.

\subsubsection{Chimera-like graphs} 
\newcommand{\C}{{\mathcal C}}
\newcommand{\CC}{\overline{\mathcal C}}
Since on small graphs we did not observe any speedup of DW compared to the classical algorithms, we now consider graphs that  fit nicely the DW architecture. The largest graph that fits DW is the Chimera graph $\mathcal C$, and since formulation \eqref{eq:mis2} uses the complement edges, the largest graph that we can solve MC on is the complement of $\C$. Let $\overline{G}$ denote the complement of any graph $G$. Note that the graphs $\C$ and $\CC$ are not interesting for the MC problem since $\C$ is bipartite and hence $\CC$ consists of two disconnected cliques, which  makes MC trivial on this graph.

Consider now the graph $\C_1$ obtained by contracting one random edge from $\C$. An \textit{edge contraction} consists of deleting an edge $(v_1,v_2)$ and merging its endpoints $v_1$ and $v_2$ into a new vertex $v^*$. With $\mathcal{N}_1$ and $\mathcal{N}_2$ the set of neighboring vertices of $v_1$ and $v_2$, the neighbors of $v^*$ are $\mathcal{N}_1\cup \mathcal{N}_2 \setminus \{v_1,v_2\} $. Solving the MC problem on $\CC_1$ requires the embedding of the complement of $\CC_1$ onto DW, which is $\C_1$. The natural embedding of $\C_1$ onto $\C$ maps $v^*$ onto a chain of two vertices and all other vertices of $\C_1$ onto single vertices of $\C$. Moreover, if we add any edge to $\C_1$, the resulting graph will not be embeddable onto $\C$ any more since $\C_1$ already uses all available qubits and edges of $\C$. We can thus say that $\C_1$ is one of the  densest graphs of size $|V|-1$ than  can be embedded onto $\C$.

We can generalize the aforementioned construction to $m$ random edge contractions; the resulting graph $\C_m$ will have $|V| -m$ vertices, will be one of the  densest graphs of size $|V| -m$ that fits onto $\C$, and the chains of such an embedding will be the paths of contracted edges. This family of graphs $\CC_m$ with $0<m<1100$ is therefore a good candidate for the best-case scenario for the MC problem: The $\CC_m$ family are large graphs whose QUBOs can be embedded onto $\C$ and whose solutions of MC are not trivial.

\begin{figure}
  \centering
  \includegraphics[width=0.6\textwidth]{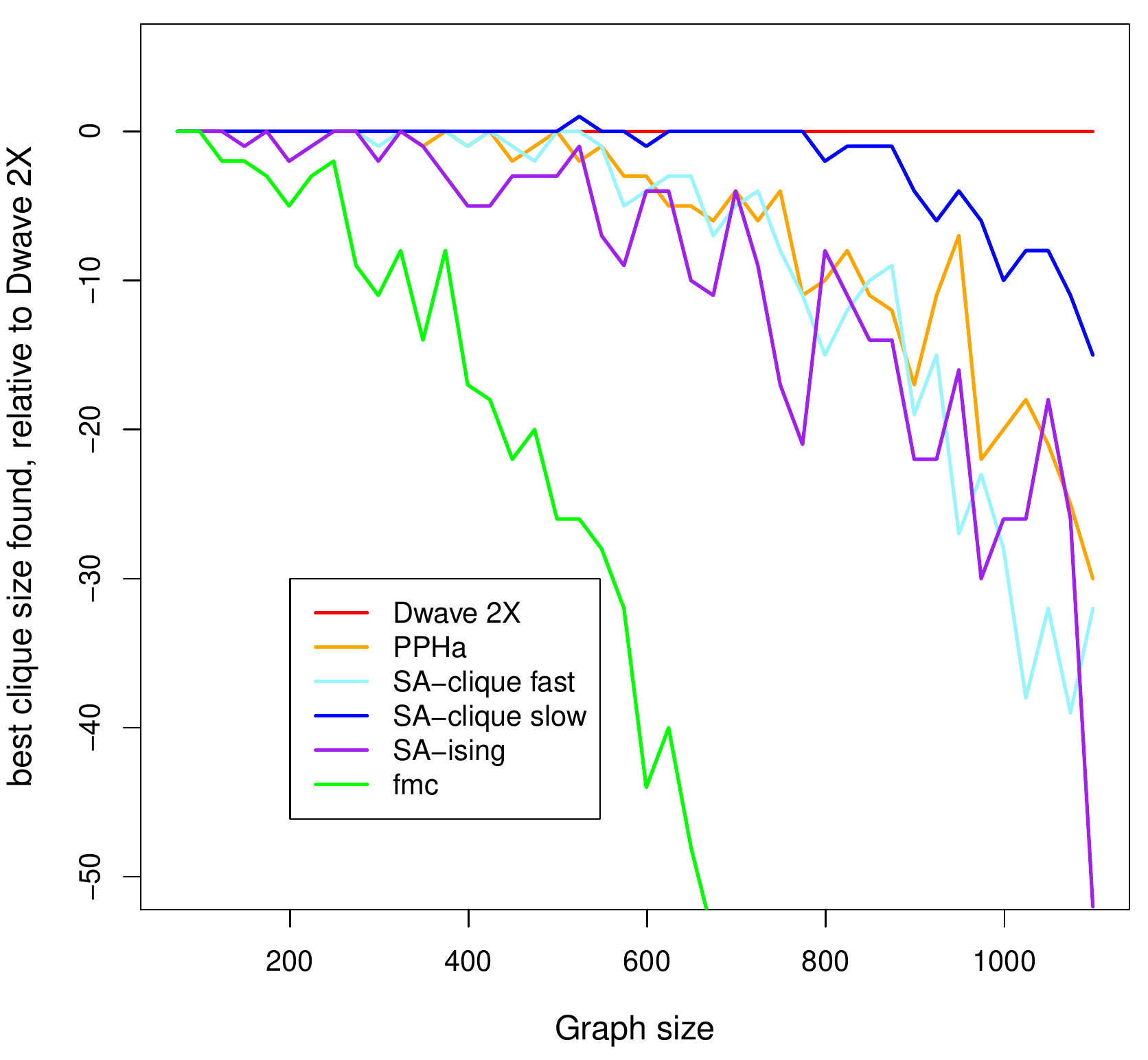}
  \caption{Best clique size found by the different solvers, relatively to the DW result, on the  $C_m$ family of graphs.}
  \label{fig:Cm}
\end{figure}

\subsubsection{Experiments}
We solve the MC problem on the $\CC_m$ family of graphs using DW's Sapi, \textit{PPHa} and the SA-Ising software, SA-clique, and \textit{fmc}.

Fig.~\ref{fig:Cm} shows the result. We observe that for graph sizes up to 400, \textit{PPHa} finds the same result as DW. For these small graphs the problem is likely simple enough to be solved by the post-processing step alone. As expected, the simulated annealing algorithms designed specifically for MC (\textit{fmc, pmc}) are behaving better than the general SA-Ising algorithm. The \textit{fmc} software is run in its heuristic mode. The comparatively lower quality results we obtain with \textit{fmc} could be due to the fact that \textit{fmc} is designed for large sparse graphs but run here on very dense graphs.

For large graphs ($\geq800$ vertices), DW gives the best solution. (Note we do not know if that solution is optimal.)

\subsubsection{Speedup}
Since SA-clique seems to be the best candidate to compete against DW, and moreover since it is considered the classical analogue of quantum annealing, we choose to compute the DW speedup relatively to SA-clique on the $\CC_m$ graph family.

We employ the following procedure: For each graph size, we run DW with 500 anneals and report the best solution. The DW runtime is the total qpu runtime for 500 anneals (approximately $0.15s$). For SA-clique, we start with a low $\alpha$ parameter (i.e., a fast cooling schedule), and gradually increase $\alpha$ until SA-clique finds the same solution as DW. The value of $\alpha$ for which SA-clique finds the same solution as DW gives us the best execution time for SA-clique given the required accuracy. The SA-clique algorithm is run on one CPU core of an Intel E8400 @ 3.00GHz.

Fig.~\ref{fig:speedup} shows the speedup for different graph sizes of the $\CC_m$ family. We observe that DW is slower than SA-clique for graphs with less than 200 vertices. For larger graphs, DW gets exponentially faster, reaching a speedup of the order of a million for graphs with $1000$ vertices. This behavior is not unexpected: For small graphs, optimized software solvers can terminate with runtimes far less than the constant anneal time of DW (see Table~\ref{tab:45g}). The larger the graphs, the more pronounced the advantage of DW is due to the fact that the $\CC_m$ graph instances investigated in this experiment are similar to the topology of DW's native Chimera graph. Further detail is given in the following section.

\begin{figure}
  \centering
  \includegraphics[width=0.6\textwidth]{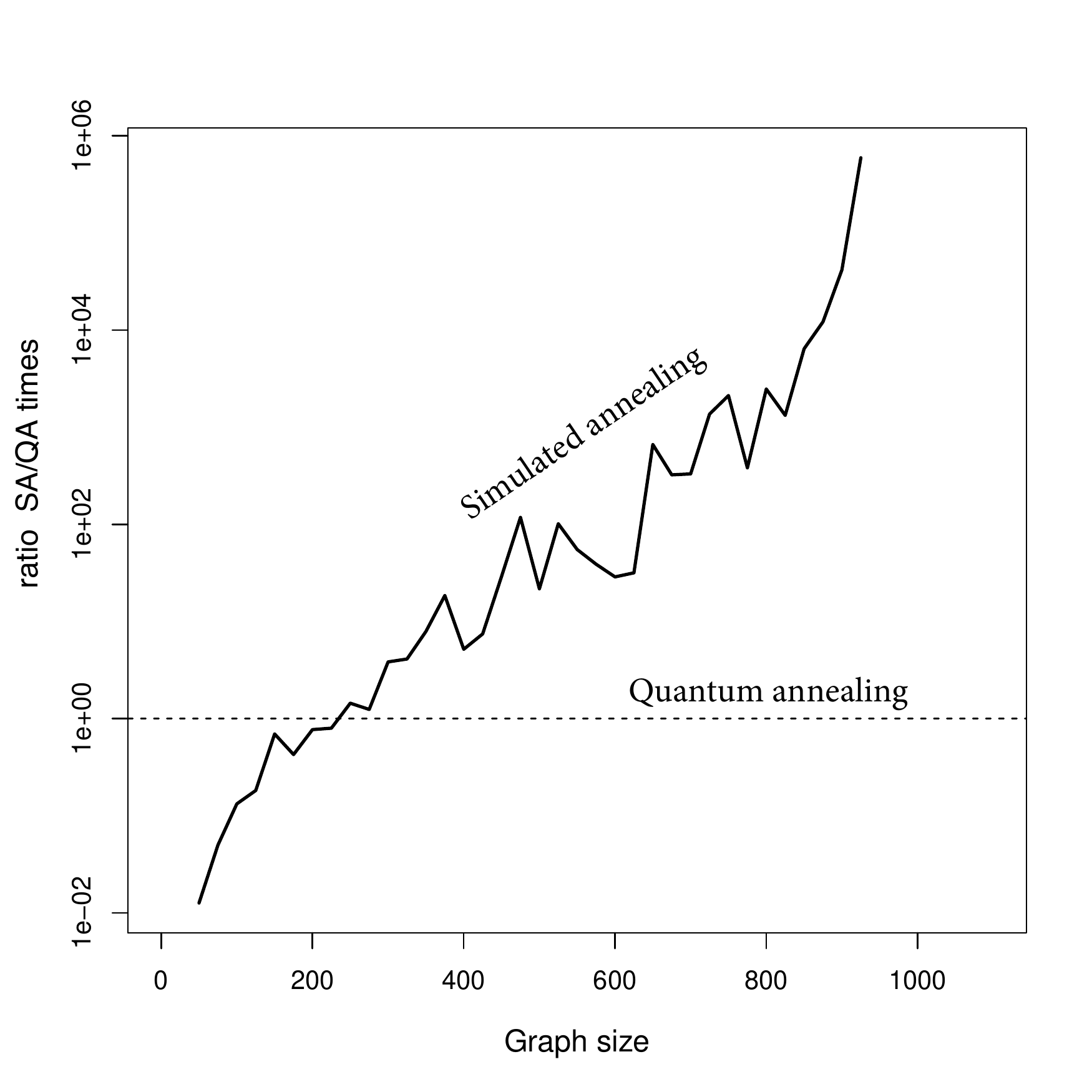}
  \caption{Speedup on artificial graphs designed to fit the Chimera topology.\label{fig:speedup}}
\end{figure}
Overall, our experiments show that for large graphs whose QUBOs can be embedded onto $\C$, DW is able to find very quickly a solution that is \emph{very} difficult to obtain with classical solvers.

\subsubsection{Topology}
In summary, the results of Sections~\ref{sec:smallgraphs} and \ref{sec:artig} demonstrate that the closer the topology of a problem is to the native Chimera graph (Fig.~\ref{fig:chimera}) of the DW chip, the more pronounced the advantage of DW over classical solvers. Moreover, with an increasing problem size, the problem becomes exponentially more difficult for classic solvers, while it takes the same time to run on DW (as long as it can be embedded onto the hardware). Note however, that the larger problem we can fit on DW (with a fixed number of qubits), the smaller average chain length we get. This means that these experiments benefit DW in the comparison with classical solvers twice: on the one hand, the problem becomes much more difficult for the classical solvers due to larger graphs involved; on the other hand, it becomes somewhat easier for DW because the shorter chains improve the accuracy, thereby biasing the results in favor of DW.

\subsection{Using decomposition for large graphs}
\label{sec:anyg}
\begin{figure}
\centering
\includegraphics[width=0.6\textwidth]{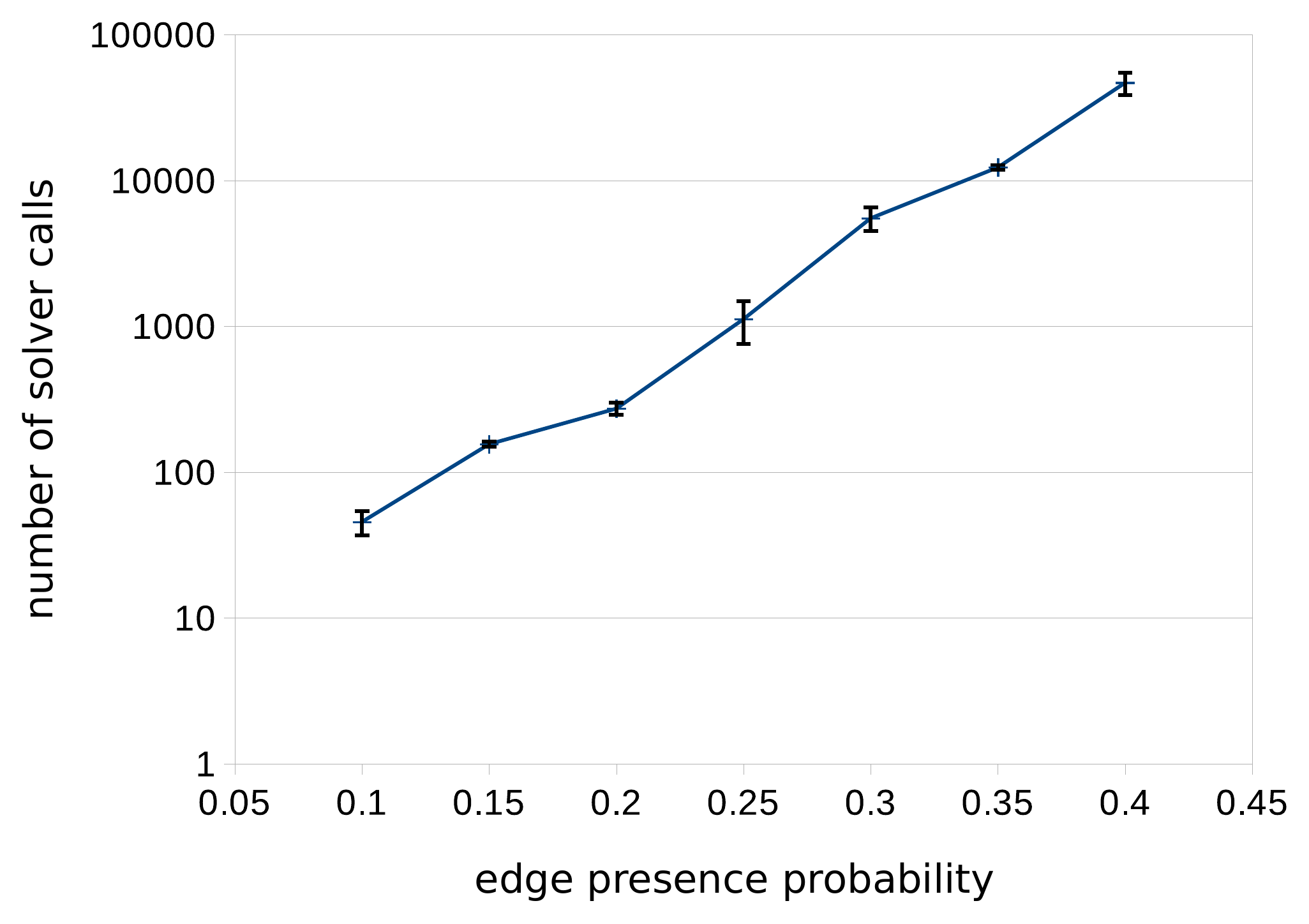}
\caption{Number of solver calls against edge probability. Log scale on the y-axis.\label{fig:largeMC_solvercalls}}
\end{figure}

We investigate some properties of the graph splitting routine of Section \ref{section_large_instances} which enables us to solve MC instances larger than the size that fits onto the DW chip. In this section, we always use our graph splitting routine to divide up the input graphs into subgraphs of $45$ vertices, the largest (complete) graphs that can be embedded on the DW chip.

First, we test our graph splitting routine on random graphs with $500$ vertices and an edge probability (edge density) ranging from $0.1$ to $0.4$ in steps of $0.05$. Fig.~\ref{fig:largeMC_solvercalls} shows the number of generated subgraphs (or equivalently, the number of solver calls) against the edge probability. Each data point is the median value of ten runs, the standard deviation is given as error bars. The number of solver calls seems to follow an exponential trend with respect to the edge probability.

\begin{figure}
  \centering
  \includegraphics[width=0.6\textwidth]{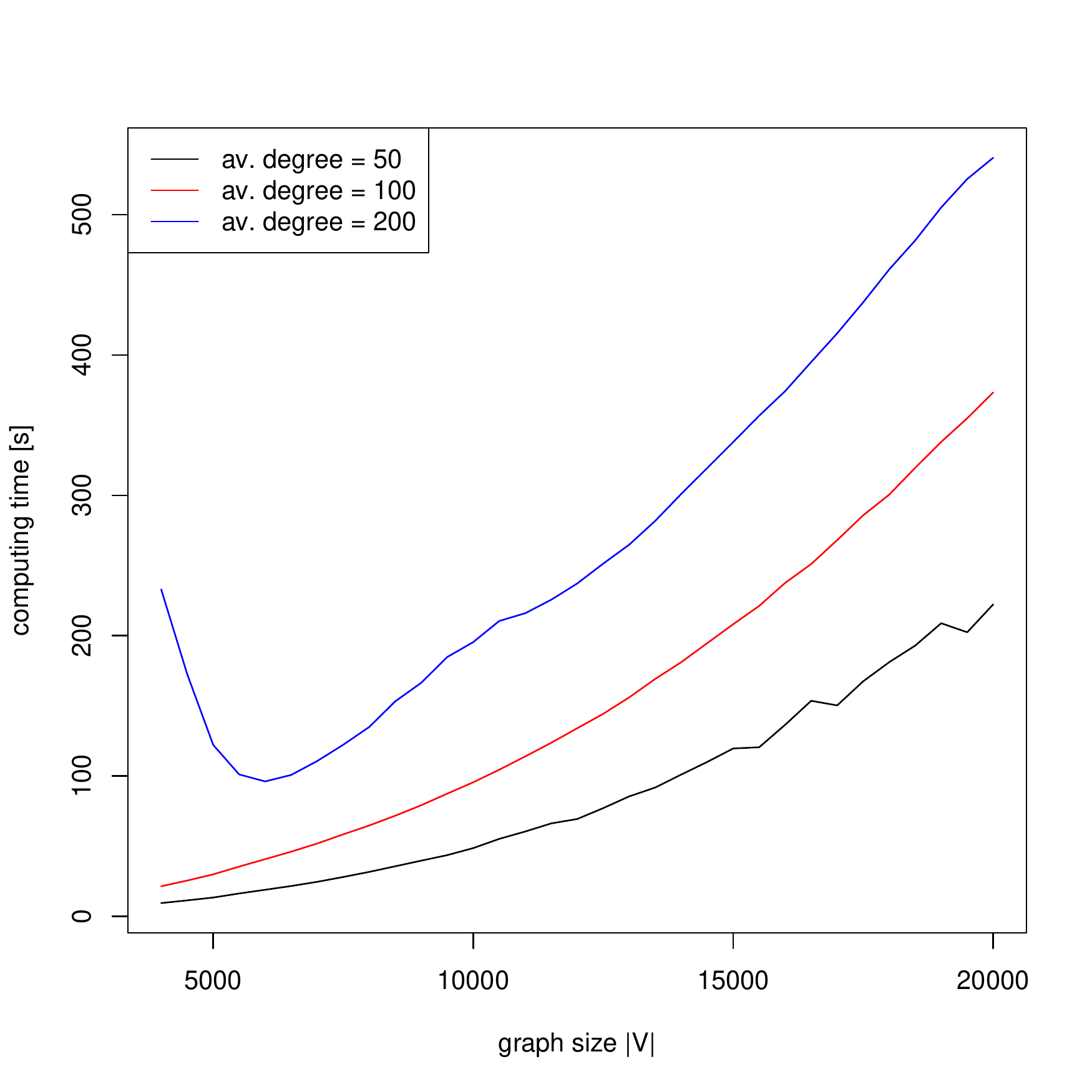}
  \caption{Time of the graph splitting routine as a function of the graph size.\label{fig:largeMC_time}}
\end{figure}

Second, we investigate the scaling of our graph splitting routine with an increasing graph size $|V|$. Since, with a fixed edge probability, graphs become denser (their vertex degrees increase) as their size goes to infinity, we take an alternative approach and fix the average degree $d$ of each vertex: We then generate graphs of size $3000$ to $20,000$ (in steps of $500$) using edge probability $p=d/(|V|-1)$. This ensures that the average vertex degree stays constant as $|V|$ goes to infinity.

We measure both the time $t$ (in seconds) of the graph splitting alone as well as the number $n$ of problems/subgraphs being solved by DW. According to Table~\ref{tab:45g} (column for DW's interface \textit{Sapi}), the time to solve each subgraph on the DW chip is $0.15$ seconds, thus leading to an overall time for computing MC of $t+0.15\cdot n$ seconds.

Fig.~\ref{fig:largeMC_time} shows average timings from $100$ runs for three fixed average degrees $d \in \{50,100,200\}$. We observe that if $d$ is relatively large in comparison to $|V|$ (which, in particular, appears to hold for $|V|\approx 5000$ and $d=200$), the $k$-core and CH-partitioning algorithms are less effective, while the vertex-splitting routine alone produces too many subgraphs,  causing the computing time to get disproportionately high. With increasing the number of vertices, we observe a roughly linear increase of the runtime. As expected, higher average degrees $d$ result in denser graphs and thus higher runtimes.

\begin{table*}
\centering
\begin{tabular}{lll lll}
graph family	& vertices	& parameter	& largest clique	& no.\ subgraphs	& runtime [s]\\
\hline\hline
Hamming	& 128	& 1	& 64	& 0	& 0.09	\\
	& 	& 2	& 32	& 196	& 6.4	\\
	&	& 4	& 4	& 20	& 0.1	\\
	&	& 6	& 2	& 1	& 0.09	\\
\hline
c-fat	& 200	& 1	& 12	& 3	& 0.02	\\
	& 	& 5	& 58	& 1	& 0.53	\\
	& 500	& 1	& 14	& 3	& 0.08	\\
	& 	& 5	& 64	& 2	& 2.1	\\
	& 	& 10	& 126	& 0	& 25.6	\\
\hline
g graph	& 100	& 10	& 1	& 1	& 0.01	\\
	& 200	& 10	& 1	& 1	& 0.01	\\
	& 500 	& 10	& 1	& 1	& 0.07	\\
	& 1000	& 10	& 1	& 1	& 0.26	\\
	& 2000	& 10	& 1	& 1	& 1.1	\\
	& 5000	& 10	& 1	& 1	& 6.8	\\
	& 10000	& 10	& 1	& 1	& 28.3	\\
\hline
U graph	& 1000	& 5	& 7	& 6	& 0.49	\\
	& 	& 10	& 10	& 9	& 0.49	\\
	& 	& 20	& 14	& 11	& 0.57	\\
	& 2000	& 5	& 7	& 7	& 1.9	\\
	& 	& 10	& 11	& 10	& 2.0	\\
	& 	& 20	& 17	& 14	& 2.2	\\
\hline\hline
\end{tabular}
\caption{Graph splitting algorithm applied to a variety of graph families (first column) including their graph parameters (number of vertices in second column, internal parameter in third column). Largest clique found, number of generated subgraphs and overall runtime in seconds for the splitting is reported.\label{tab:graphfamilies}}
\end{table*}

To demonstrate the applicability of our graph splitting routine outside of random graphs, we apply the graph splitting to families of graphs from the \textit{1993 DIMACS Challenge on Cliques, Coloring and Satisfiabilty} \citep{johnson-trick-96}, also used in \cite{Boros2006}. These are \textit{Hamming} and \textit{c-fat} graphs. Both graph families depend on two parameters: the number of vertices $n$ and an additional internal parameter, the Hamming distance $d$ for Hamming graphs and the partition parameter $c$ for c-fat graphs. We use the generation algorithms of \cite{Hasselberg1993} for both graph families. We also employed $g$ and $U$ graphs, defined in \cite{Kim2001} (including their generation mechanism), which have previously been used for graph assessments in \cite{Kim2001,Boros2006}.

Table~\ref{tab:graphfamilies} shows results for all four graph families. We see that for the graph parameters used in the aforementioned studies, our graph splitting algorithm finds a maximum clique (mostly) within a fraction of a second. The number of generated subgraphs along the way varies widely, from none or one subgraph for $g$ graphs to almost two hundred for Hamming graphs.

\begin{figure}
  \centering
  \includegraphics[width=0.7\textwidth]{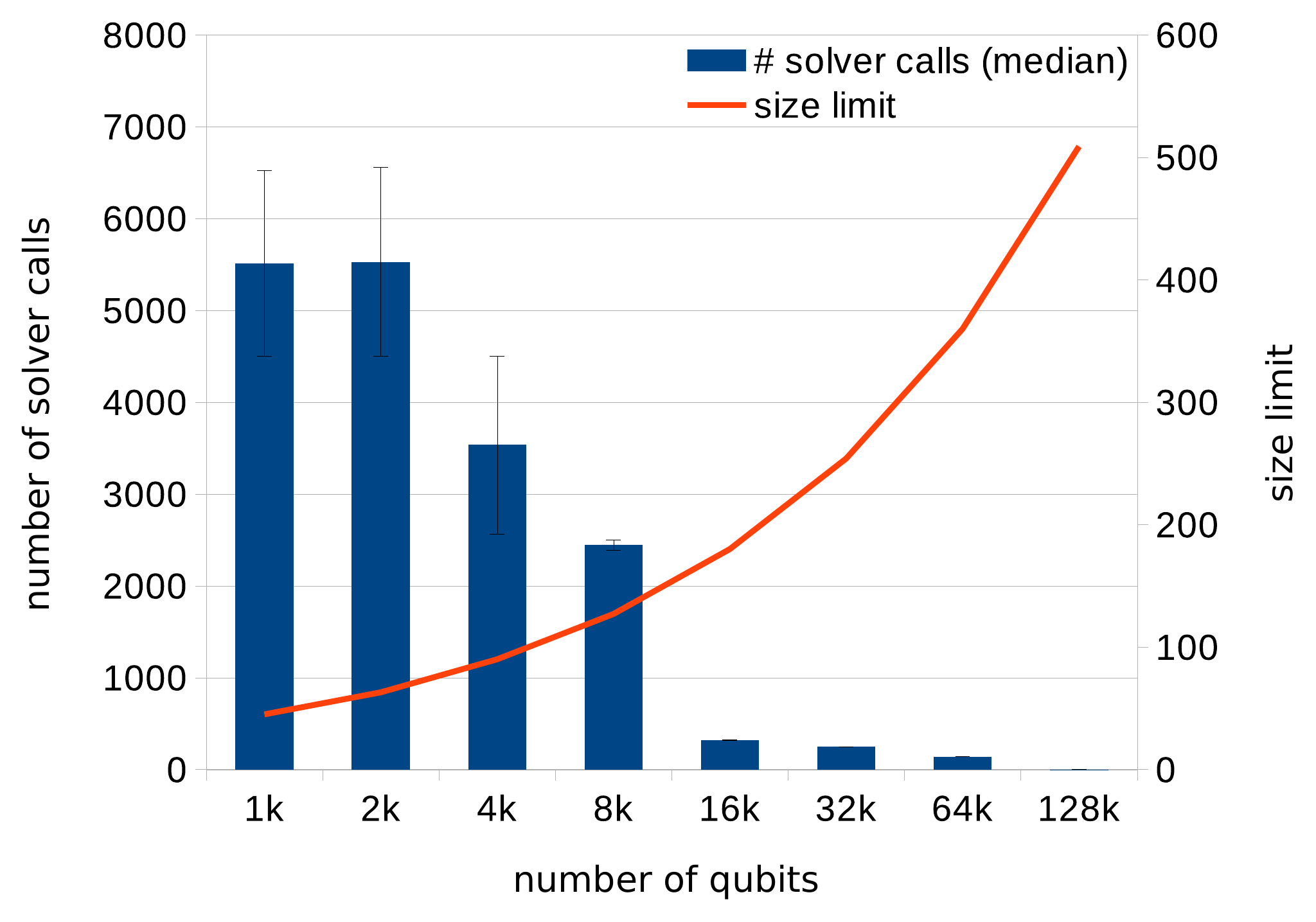}
  \caption{Number of solver calls (left y-axis) and size limit (maximal arbitrary graph embeddable on DW; right y-axis) as a function of the number of qubits.\label{fig:number-of-solver}}
\end{figure}

Lastly, we aim to assess the performance of future generations of DW systems on our clique finding approach for arbitrary large graphs. Essentially, we turn the previous question around: Instead of assessing the graph splitting for a variety of graphs and a fixed DW system, we now look at the evolution of possible future DW machines with an increasing number of qubits and investigate the number of solver calls needed by the graph splitting algorithm (applied to a fixed realization of a random graph with $500$ vertices and edge presence probability $0.3$).

First, assuming a similar Chimera topology for future generations of DW systems, doubling the number of available qubits will increase the size of the maximal complete subgraph that can be embedded by a factor of $\sqrt{2}$. The maximal size of an arbitrary graph embeddable on DW is shown in Fig.~\ref{fig:number-of-solver} in red (right y-axis). If we assume that the number of qubits doubles with each new generation, seven generations of DW machines are required in order to be able to directly embed and solve an arbitrary $500$ vertex graph.

Second, Fig.~\ref{fig:number-of-solver} (blue data line; left y-axis) shows the evolution of the number of solver calls for future DW systems with an increasing number of qubits. We use the envisaged size of the maximal complete subgraph embeddable on future DW machines to set the \textit{lower\_bound} parameter of the graph splitting algorithm. In this experiment we applied the graph splitting algorithm to the fixed graph generated with $500$ vertices and edge presence probability of $0.3$. Each data point is the median of ten runs. The standard deviation of those ten runs is given with error bars. The number of required solver calls of our graph splitting algorithm rapidly decreases in what seems like an exponential trend.

\section{Conclusion}
\label{section_conclusion}
This article evaluates the performance of the DW quantum annealer on maximum clique, an important NP-hard graph problem. We compared DW's solvers to common classical solvers with the aim of determining if current technology already allows us to observe a \textit{quantum advantage} for our particular problem. We summarize our findings as follows.
\begin{enumerate}
  \item The present DW chip capacity of around $1000$ qubits poses a significant limitation on the MC problem instances of general form that can be solved directly with DW. For random graphs with no special structure that are small enough to fit onto DW, the returned solution is of comparable quality to the one obtained by classical methods. Nevertheless the highly optimized classical solvers available are usually faster for such small instances.
  \item Special instances of large graphs designed to fit DW's chimera architecture can be solved orders of magnitude faster with DW than with any classical solvers.
  \item For MC instances that do not fit DW, the proposed decomposition methods offer a way to divide the MC problem into subproblems that fit DW. The solutions of all subproblems can be combined afterwards into an optimal solution of the original problem (assuming DW solves the subproblems optimally, which is usually true, but cannot be guaranteed). Our decomposition methods are highly effective for relatively sparse graphs; however the number of subproblems generated grows exponentially with increasing density. We demonstrate that this issue can be alleviated when/if larger D-Wave machines become available (Fig.~\ref{fig:number-of-solver}).
\end{enumerate}
Overall, we conclude that general problem instances that allow to be mapped onto the DW architecture are typically still too small to show a quantum advantage. But quantum annealing may offer a significant speedup for solving the MC problem, if the problem size is at least several hundred, roughly an order of magnitude larger than what it typically is for general problems that fit D-Wave 2X.

\section*{Acknowledgments}
The authors acknowledge and appreciate the support provided for this work by the Los Alamos National Laboratory Directed Research and Development Program (LDRD). They would also like to thank Dr Denny Dahl for his help while working on the D-Wave 2X machine.

% \bibliographystyle{apalike}
% \bibliography{quantum}

\begin{thebibliography}{}
\bibitem[Balas and Yu, 1986]{Balas1986}
Balas, E. and Yu, C. (1986).
\newblock Finding a maximum clique in an arbitrary graph.
\newblock {\em SIAM J Comp}, 15:1054--1068.

\bibitem[Batagelj and Zaversnik, 2011]{DBLP:journals/corr/cs-DS-0310049}
Batagelj, V. and Zaversnik, M. (2011).
\newblock An o(m) algorithm for cores decomposition of networks.
\newblock {\em Adv Data Anal Classi}, 5(2):129--145.

\bibitem[Boothby et~al., 2016]{Boothby2016}
Boothby, T., King, A., and Roy, A. (2016).
\newblock {Fast clique minor generation in Chimera qubit connectivity graphs}.
\newblock {\em Quantum Information Processing}, 15(1):495--508.

\bibitem[Boros et~al., 2006]{Boros2006}
Boros, E., Hammer, P., and Tavares, G. (2006).
\newblock {Preprocessing of Unconstrained Quadratic Binary Optimization}.
\newblock {\em Rutcor Research Report}, RRR 10-2006:1--58.

\bibitem[Bunyk et~al., 2014]{Bunyk2014}
Bunyk, P., Hoskinson, E., Johnson, M., Tolkacheva, E., Altomare, F., Berkley,
  A., Harris, R., Hilton, J., Lanting, T., Przybysz, A., and Whittaker, J.
  (2014).
\newblock Architectural considerations in the design of a superconducting
  quantum annealing processor.
\newblock {\em IEEE Trans on Appl Superconductivity}, 24(4):1--10.

\bibitem[Cao et~al., 2016]{Cao2016}
Cao, Y., Jiang, S., Perouli, D., and Kais, S. (2016).
\newblock {Solving Set Cover with Pairs Problem using Quantum Annealing}.
\newblock {\em Nature Scientific Reports}, 6(33957):1--15.

\bibitem[Coffrin et~al., 2017]{Coffrin2017}
Coffrin, C., Nagarajan, H., and Bent, R. (2017).
\newblock Challenges and successes of solving binary quadratic programming
  benchmarks on the dw2x qpu.
\newblock {\em Los Alamos ANSI debrief}, pages 1--84.

\bibitem[{D-Wave}, 2016a]{dwavepostprocess24}
{D-Wave} (2016a).
\newblock {D-Wave} post-processing guide.

\bibitem[{D-Wave}, 2016b]{dwave2016}
{D-Wave} (2016b).
\newblock Introduction to the {D-Wave} quantum hardware.

\bibitem[Denchev et~al., 2016]{Denchev2016}
Denchev, V., Boixo, S., Isakov, S., Ding, N., Babbush, R., Smelyanskiy, V.,
  Martinis, J., and Neven, H. (2016).
\newblock What is the computational value of finite-range tunneling?
\newblock {\em Phys Rev X}, 6:031015.

\bibitem[Djidjev et~al., 2016]{Djidjev2016}
Djidjev, H., Hahn, G., Mniszewski, S., Negre, C., Niklasson, A., and
  Sardeshmukh, V. (2016).
\newblock Graph partitioning methods for fast parallel quantum molecular
  dynamics.
\newblock {\em CSC 2016}, 1(1):1--17.

\bibitem[Dridi and Alghassi, 2016]{DridiAlghassi2016}
Dridi, R. and Alghassi, H. (2016).
\newblock Homology computation of large point clouds using quantum annealing.
\newblock {\em arXiv:1512.09328}, pages 1--17.

\bibitem[Geng et~al., 2007]{geng2007simple}
Geng, X., Xu, J., Xiao, J., and Pan, L. (2007).
\newblock A simple simulated annealing algorithm for the maximum clique
  problem.
\newblock {\em Inf Sciences}, 177(22):5064--5071.

\bibitem[{Gurobi Optimization, Inc.}, 2015]{gurobi}
{Gurobi Optimization, Inc.} (2015).
\newblock Gurobi optimizer reference manual.

\bibitem[Hasselberg et~al., 1993]{Hasselberg1993}
Hasselberg, J., Pardalos, P., and Vairaktarakis, G. (1993).
\newblock {Test Case Generators and Computational Results for the Maximum
  Clique Problem}.
\newblock {\em Journal of Global Optimization}, 3:463--482.

\bibitem[Johnson and Trick, 1996]{johnson-trick-96}
Johnson, D.~S. and Trick, M.~A., editors (1996).
\newblock {\em Clique, Coloring, and Satisfiability: Second {DIMACS}
  Implementation Challenge, {DIMACS}}, volume~26.
\newblock American Mathematical Society.

\bibitem[Johnson et~al., 2011]{Johnson2011}
Johnson, M., Amin, M., Gildert, S., Lanting, T.~Hamze, F., Dickson, N., Harris,
  R., Berkley, A., Johansson, J., Bunyk, P., Chapple, E., Enderud, C., Hilton,
  J., Karimi, K., Ladizinsky, E., Ladizinsky, N., Oh, T., Perminov, I., Rich,
  C., Thom, M., Tolkacheva, E., Truncik, C., Uchaikin, S., Wang, J., B., W.,
  and Rose, G. (2011).
\newblock Quantum annealing with manufactured spins.
\newblock {\em Nature}, 473:194--198.

\bibitem[Kim et~al., 2001]{Kim2001}
Kim, S.-H., Kim, Y.-H., and Moon, B.-R. (2001).
\newblock {A Hybrid Genetic Algorithm for the MAX CUT Problem}.
\newblock {\em Proceeding GECCO'01 Proceedings of the 3rd Annual Conference on
  Genetic and Evolutionary Computation}, pages 416--423.

\bibitem[King et~al., 2015]{King2015}
King, J., Yarkoni, S., Nevisi, M., Hilton, J., and McGeoch, C. (2015).
\newblock Benchmarking a quantum annealing processor with the time-to-target
  metric.
\newblock {\em arXiv:1508.05087}, pages 1--29.

\bibitem[Lucas, 2014]{Lucas2014}
Lucas, A. (2014).
\newblock Ising formulations of many np problems.
\newblock {\em Frontiers in Physics}, 2(5):1--27.

\bibitem[Mniszewski et~al., 2016]{Mniszewski2016}
Mniszewski, S., Negre, C., and Ushijima-Mwesigwa, H. (2016).
\newblock {Graph Partitioning using the D-Wave for Electronic Structure
  Problems}.
\newblock {\em LA-UR-16-27873}, pages 1--21.

\bibitem[Neukart et~al., 2017]{Neukart2017}
Neukart, F., Von~Dollen, D., Compostella, G., Seidel, C., Yarkoni, S., and
  Parney, B. (2017).
\newblock Traffic flow optimization using a quantum annealer.
\newblock {\em arXiv:1708.01625}, pages 1--12.

\bibitem[Nguyen and Kenyon, 2017]{NguyenKenyon2017}
Nguyen, N. and Kenyon, G. (2017).
\newblock Solving sparse representations for object classification using the
  quantum d-wave 2x machine.
\newblock {\em Los Alamos ISTI debrief}, pages 1--30.

\bibitem[Pattabiraman et~al., 2013]{pattabiraman2013fast}
Pattabiraman, B., Patwary, M., Gebremedhin, A., Liao, W.-K., and Choudhary, A.
  (2013).
\newblock Fast algorithms for the maximum clique problem on massive sparse
  graphs.
\newblock In {\em International Workshop on Algorithms and Models for the
  Web-Graph}, pages 156--169. Springer.

\bibitem[Perdomo-Ortiz et~al., 2017]{PerdomoOrtiz2017}
Perdomo-Ortiz, A., Feldman, A., Ozaeta, A., Isakov, S., Zhu, Z., O’Gorman,
  B., Katzgraber, H., Diedrich, A., Neven, H., de~Kleer, J., Lackey, B., and
  Biswas, R. (2017).
\newblock On the readiness of quantum optimization machines for industrial
  applications.
\newblock {\em arXiv:1708.09780}, pages 1--22.

\bibitem[Rogers and Singleton, 2016]{RogersSingleton2016}
Rogers, M. and Singleton, R. (2016).
\newblock {Ising Simulations on the D-Wave QPU}.
\newblock {\em LA-UR-16-27649}, pages 1--14.

\bibitem[R{\o}nnow et~al., 2014]{Ronnow2014}
R{\o}nnow, T.~Wang, Z., Job, J., Boixo, S., Isakov, S., Wecker, D., Martinis,
  J., Lidar, D., and Troyer, M. (2014).
\newblock Defining and detecting quantum speedup.
\newblock {\em Science}, 345:420--424.

\bibitem[Rossi et~al., 2013]{rossi2013fast}
Rossi, R., Gleich, D., Gebremedhin, A., and Patwary, M. (2013).
\newblock A fast parallel maximum clique algorithm for large sparse graphs and
  temporal strong components.
\newblock {\em CoRR, abs/1302.6256}.

\bibitem[Stollenwerk et~al., 2015]{Stollenwerk2015}
Stollenwerk, T., Lobe, E., and Tr\"oltzsch, A. (2015).
\newblock Discrete optimisation problems on an adiabatic quantum computer.
\newblock London, England. 17th British-French-German Conference on
  Optimization.

\bibitem[Thulasidasan, 2016]{Thulasidasan2016}
Thulasidasan, S. (2016).
\newblock {Generative Modeling for Machine Learning on the D-Wave}.
\newblock {\em LA-UR-16-28813}, pages 1--23.

\bibitem[Trummer and Koch, 2015]{TrummerKoch2015}
Trummer, I. and Koch, C. (2015).
\newblock {Multiple Query Optimization on the D-Wave 2X adiabatic Quantum
  Computer}.
\newblock {\em arXiv:1510.06437}, pages 1--12.

\bibitem[Ushijima-Mwesigwa et~al., 2017]{Ushijima2017}
Ushijima-Mwesigwa, H., Negre, C., and Mniszewski, S. (2017).
\newblock {Graph Partitioning using Quantum Annealing on the D-Wave System}.
\newblock {\em arXiv:1705.03082}, pages 1--20.

\bibitem[Wang et~al., 2016]{Wang2016}
Wang, C., Chen, H., and Jonckheere, E. (2016).
\newblock Quantum versus simulated annealing in wireless interference network
  optimization.
\newblock {\em Nature Scientific Reports}, 6(25797):1--9.
\end{thebibliography}

\end{document}